\documentclass[sigconf,nonacm,screen]{acmart}

\usepackage{prelude}

\let\origfootnote\footnote
\renewcommand{\footnote}[1]{\kern.06em\origfootnote{#1}}
\newcommand{\pfootnote}[1]{\kern-.06em\origfootnote{#1}}

\begin{document}
\title{Retargeting an Abstract Interpreter for a New Language by Partial Evaluation}

\author{Jay Lee}
\orcid{0000-0002-2224-4861}
\email{jhlee@ropas.snu.ac.kr}

\affiliation{%
  \department{Department of Computer Science and Engineering}
  \institution{Seoul National University}
  \city{Seoul}
  \country{Korea}}


\maketitle

\section{Problem and Motivation}\label{sec:prob}
It is well-known that abstract interpreters \citep{CouCou77Abstract} can be systematically derived from their concrete counterparts using a ``recipe,'' \citep{RivYi20Introduction} but developing sound static analyzers remains a time-consuming task.
Reducing the effort required and mechanizing the process of developing analyzers continues to be a significant challenge.
Is it possible to automatically retarget an existing abstract interpreter for a new language?

We propose a novel technique to automatically derive abstract interpreters for various languages from an existing abstract interpreter.
By leveraging partial evaluation \citep{Fut71Partial,Fut99Partial}, we specialize an abstract interpreter for a source language.
The specialization is performed using the semantics of target languages written in the source language \citep{Rey72Definitional}.
Our approach eliminates the need to develop analyzers for new targets from scratch.
We show that our method can effectively retarget an abstract interpreter for one language into a correct analyzer for another language.

\section{Background and Related Work}\label{sec:back}
\paragraph{Partial Evaluation}
Partial evaluation specializes an interpreter with respect to a program, effectively compiling it---an observation known as the first \citet{Fut71Partial,Fut99Partial} projection.
The result is a specialized interpreter that only accepts the input for the program it was specialized for.
Interpreting a program with different inputs repeatedly incurs the overhead of traversing the program's AST; Specializing the interpreter eliminates this overhead \citep{Jon96Introduction}.

Staged abstract interpreters \citep{WeiCheRom19Staged} extend the first Futamura projection to monadic abstract interpreters in the style of Abstracting Definitional Interpreters \citep{DarLabNguVan17Abstracting}.
This approach demonstrates performance improvements when analyzing programs by specializing the abstract interpreter for each target.

Our approach differs fundamentally, as we specialize an abstract interpreter not for an analysis target, but for a \emph{definitional interpreter} for a new language.

\paragraph{Meta-level Analysis}
\textsf{JSAVER} \citep{ParAnRyu22Automatically} is a meta-level static analyzer that analyzes JavaScript\,(JS) programs according to the rapidly evolving ECMA-262 \citep{ECM15ECMAScript} specification.
It analyzes JS programs by performing meta-level analysis on a definitional interpreter for JS written in a lower-level language called~\IRES.
In this specialized setting, JS and \IRES{} share the same base values, and the analysis operates on the AST of the JS interpreter.

Our approach offers a promising extension of this in two ways:
\begin{enumerate}
  \item Our method is language-agnostic, allowing source and target languages to have distinct value domains.
  \item We address the performance limitations identified in their work through partial evaluation.
\end{enumerate}

\paragraph{Reusing Abstractions}
Our motivation aligns with the broader concept of reusing abstractions in static analysis.
Mechanized frameworks such as skeletal semantics \citep{BodGarJenSch19Skeletal} have been proposed to derive abstract interpreters by integrating meta-language abstractions with language-specific ones \citep{JenRebSch23Deriving}.
Another approach by \citet{KeiErd19Sound} employs arrows meta-language \citep{Hug00Generalising} to compose modular, reusable analysis components.
\textsc{Mopsa} \citep{JouMinMonOua20Combinations}, an open-source static analysis framework, leverages OCaml's extensible variants to easily define new targets and reuse abstractions compositionally.

Unlike these bottom-up approaches that compose analysis components, our method reuses an existing abstract interpreter in its entirety to derive a new one.

\section{Approach}\label{sec:approach}
\begin{nota}
  We denote the meta-language, the source language, and the target language as~\MET,~\SRC, and~\TGT, respectively.
  Programs and values in these languages are typeset differently, e.g.,~$\mathM{e}$,~$\mathS{e}$, and~$\mathT{e}$.
  When it is helpful, we subscript the language name to distinguish the objects, e.g.,~$\Dom ME$, $\Dom SE$, and~$\Dom TE$.
  Given domains~$\mathbb{D}_1$ and~$\mathbb{D}_2$, we denote the product as~$\mathbb{D}_1 \times \mathbb{D}_2$, disjoint union as~$\mathbb{D}_1 + \mathbb{D}_2$, and continuous function domain as~$\mathbb{D}_1 \domto \mathbb{D}_2$, where~$\times$ has the highest precedence, $+$ has lower precedence, and~$\domto$ the lowest; $\times$ and~$+$ are left-associative, whereas~$\domto$ is right-associative~\citep{Rey98Theories}.
  We denote the set of partial functions from~$X$ to~$Y$ as~$X \parto Y$.
\end{nota}

Our goal is
\begin{description}
  \item[given] a concrete \TGT*interpreter $\Intp ST$ in~\SRC\,(\cref{subsec:int}), and
  \item[given] an abstract \SRC*interpreter $\Intp MS<\ABS>$ in~\MET\,(\cref{subsec:aint}), to
  \item[derive] an abstract \TGT*interpreter $\Intp MT<\ABS>$ in~\MET\,(\cref{subsec:correctness}).
\end{description}

We want the derivation to be \emph{fully automatic,} hence avoiding the need to develop an abstract interpreter for~\TGT{} from scratch.
As we shall see in \cref{subsec:peval}, we achieve this by partial evaluation~$\PE$ in the \MET*language.
In \cref{subsec:correctness}, we show that the retargeted abstract interpreter is indeed sound:
\OmitEnc
\OmitDec
\begin{restatable*}[Correctness of a Retargeted Abstract Interpreter]{thm}{correctness}\label{thm:correctness}
  Retargeting an abstract interpreter~$\Intp*MS<\ABS>$ with respect to an interpreter~$\Intp ST$ results in a sound abstract interpreter~$\Intp*MT<\ABS> \triangleq \Pe[\big]{\Intp MS<\ABS>}{\Enc[\big]SM{\Intp ST}}$.
\end{restatable*}
\noindent We finish by presenting a mini-example in \cref{subsec:example} to illustrate our approach.

\IncludeEnc
\IncludeDec
We now define the ingredients of our recipe.
\begin{defn}[Syntax and Semantics]\label{def:synsem}
  Let~$\Dom ME$,~$\Dom SE$, and~$\Dom TE$ be the sets of expressions, where~$\Dom ME$ and~$\Dom TE$ contain tuple constructors
  \begin{equation*}
    \Dom ME \ni \mathM{e} \Coloneqq \mathM{(\mathM{e, e})} \mathbin| \dotsb \quad \text{and} \quad
    \Dom SE \ni \mathS{e} \Coloneqq \mathS{(\mathS{e, e})} \mathbin| \dotsb
  \end{equation*}
  where $\dotsb$ denotes the rest of the syntax, which is parameterizable.

  Let~$\Dom MV$,~$\Dom SV$, and~$\Dom TV$ be the domains of values for languages~\MET, \SRC, and~\TGT, respectively.
  Then we have concrete semantics functions
  \begin{equation*}
    \Sem M{e} \in \Dommath M{\Sigma} \domto \Dom MV, \qquad
    \Sem S{e} \in \Dommath S{\Sigma} \domto \Dom SV,\text{ and } \quad
    \Sem T{e} \in \Dommath T{\Sigma} \domto \Dom TV
  \end{equation*}
  for all programs~$\mathM{e} \in \Dom ME$, $\mathS{e} \in \Dom SE$, and~$\mathT{e} \in \Dom TE$, where $\Dommath M{\Sigma}$, $\Dommath S{\Sigma}$, and $\Dommath T{\Sigma}$ are (parameterizable) environments.
  For simplicity in our presentation, we assume that environments are single-valued, i.e.,~$\Dommath M{\Sigma} = \Dom MV$,~$\Dommath S{\Sigma} = \Dom SV$, and~$\Dommath T{\Sigma} = \Dom TV$.
  These essentially represent the input values of the programs.
  Domains~$\Dom MV$ and~$\Dom SV$ should be closed under the product operation, i.e., $\Dom MV \times \Dom MV \subseteq \Dom MV$ and $\Dom SV \times \Dom SV \subseteq \Dom SV$.
  Angle brackets denote (semantic) tuples, e.g., $\mathM{\Tup{v_1, v_2}}$ and $\mathS{\Tup{v_1, v_2}}$.
\end{defn}
An additional requirement for each domain of defining-language values is that it should be able to encode the programs and the values of the defined-language, and we express this in \cref{def:encdec}.

\subsection{An Interpreter for the New Language}\label{subsec:int}
The first ingredient is the semantics of the new language~\TGT{}: a definitional interpreter~$\Intp ST$.
This is written in the source language~\SRC{} for which we have an abstract interpreter~$\Intp*MS<\ABS>$.

For an interpreter to accept a program of another target language, we need to embed programs and values of the target language into the source language.
\begin{defn}[Encoders and Decoders]\label{def:encdec}
  An encoder~$\Enc TS{} \in \Dom TE + \Dom TV \domto \Dom SV$ is a function that encodes a \TGT*program or a \TGT*value into an \SRC*expression.
  This induces a decoder~$\Dec ST{} \in \Dom SV \parto \Dom TE + \Dom TV$, which is a partial function, where
  \begin{equation*}
    \Dec[\big]ST{\Enc TS{e}} = \mathT{e} \quad \text{and} \quad \Dec[\big]ST{\Enc TS{v}} = \mathT{v}
  \end{equation*}
  for all $\mathT{e} \in \Dom TE$ and $\mathT{v} \in \Dom TV$.
\end{defn}
In practice, encoders and decoders are usually part of an interpreter as parsers, and we can assume that they are available.
For brevity, we shall omit the encoders and decoders as they can be unambiguously inferred from the surrounding context.

\OmitEnc
\OmitDec

\begin{defn}[Concrete Interpreter]\label{def:int}
  A concrete \TGT*interpreter~$\Intp ST \in \Dom SE$ is an \SRC*program that accepts an \SRC*tuple that encodes a \TGT*program and a \TGT*value, and returns an \SRC*value that encodes the evaluation result.
  That is, $\Intp ST$ satisfies
  \begin{equation}\label{eq:int}
    \Sem[\big]S{\Intp ST} \mathS{\Tup{\Enc TS{e}, \Enc TS{i}}} = \Enc[\big]TS{\Sem T{e}\: i}
  \end{equation}
  for all $\mathT{e} \in \Dom TE$ and $\mathT{i} \in \Dom TV$.
\end{defn}

Note that \cref{eq:int} enforces a few requirements on the expressiveness of the domains.
\SRC*language should to be able to construct and represent a tuple, as well as to encode the programs and values of \TGT.
And of course, \SRC*language should be able to write an interpreter for \TGT.
These explain the requirements in \cref{def:synsem}.

\subsection{An Existing Abstract Interpreter}\label{subsec:aint}
Now we need the second ingredient: an existing abstract interpreter for~\SRC,~$\Intp*MS<\ABS>$.
This is written in the meta-language~\MET{} for which we have a partial evaluator~$\PE$.

We define the abstract domain and the concretization function for the source language~\SRC{} in which the abstract interpreter~$\Intp*MS<\ABS>$ operates.
\begin{defn}[Abstract Domain and Concretization]\label{def:conc}
  The abstract domain~$\Dom*SV<\ABS>$ of~\SRC{} approximates the concrete domain~$\Dom SV$ with an approximation partial order~$\leS$.
  An abstract value can be concretized to a set of concrete values using a monotonic concretization function~$\Conc S \in \Dom*SV<\ABS> \domto \Pow{\Dom SV}$ where for any $\mathS{v_1^\ABS}, \mathS{v_1^\ABS} \in \Dom*SV<\ABS>$, we have \[\mathS{v_1^\ABS} \leS \mathS{v_2^\ABS} \implies \Conc S \: \mathS{v_1^\ABS} \subseteq \Conc S \: \mathS{v_2^\ABS}. \qedhere\]
\end{defn}

With \cref{def:conc}, we can define a correctness condition for an abstract \SRC*interpreter~$\Intp MS<\ABS>$.
\begin{defn}[Abstract Interpreter]\label{def:aint}
  An abstract \SRC*interpreter~$\Intp*MS<\ABS> \in \Dom ME$ is an \MET*program that accepts an \SRC*program and an \SRC*value, and returns a sound approximation of the evaluation result with respect to the concretization~$\Conc S \in \Dom*SV<\ABS> \domto \Pow{\Dom SV}$, as defined in \cref{def:conc}.
  That is, $\Intp*MS<\ABS>$~satisfies
  \begin{equation}\label{eq:aint}
    \Enc[\big]SM{\Sem S{e}\: \mathS{i}} \in \Conc S \Comp \Sem[\big]M{\Intp MS<\ABS>} \: \Tup{\Enc SM{e}, \Enc SM{i}}
  \end{equation}
  for all~$\mathS{e} \in \Dom SE$ and~$\mathS{i} \in \Dom SV$.
\end{defn}

Again, \cref{eq:aint} imposes requirements on the expressiveness of \MET.
The meta-language~\MET{} should be able to construct and represent a tuple, encode the programs and values of \SRC, and write an abstract interpreter for \SRC.

\subsection{Retargeting using a Partial Evaluator}\label{subsec:peval}
The final ingredient for automatically deriving a retargeted abstract interpreter is a partial evaluator for \MET: $\PE \in \Dom ME \times \Dom MV \domto \Dom ME$.
\begin{defn}[Partial Evaluator]\label{def:peval}
  A partial evaluator~$\PE \in \Dom ME \times \Dom MV \domto \Dom ME$ accepts two inputs,
  \begin{enumerate}
    \item an \MET*program~$\mathM{e} \in \Dom ME$ that accepts an \MET*tuple~$\mathM{\Tup{i_1, i_2}} \in \Dom MV$, and
    \item an \MET*value~$\mathM{i_1} \in \Dom MV$,
  \end{enumerate}
  and returns a specialized \MET*program that accepts the remaining \MET*value~$\mathM{i_2} \in \Dom MV$.
  That is, $\PE$~satisfies
  \begin{equation}\label{eq:peval}
    \Sem M{\Pe[\big]{e}{i_1}}\: \mathM{i_2} = \Sem M{e} \mathM{\Tup{i_1, i_2}}
  \end{equation}
  for all~$\mathM{e} \in \Dom ME$ and~$\mathM{i_1}, \mathM{i_2} \in \Dom MV$.
\end{defn}

From \cref{eq:peval}, we can specialize an abstract interpreter with a target program by substituting~$\Intp MS<\ABS>$ for~$\mathM{e}$:
\begin{equation}\label{eq:fut}
  \Sem[\big]M{\Pe[\big]{\Intp MS<\ABS>}{\Enc SM{e}}}\: \Enc SM{i}
  = \Sem[\big]M{\Intp MS<\ABS>}\: \Tup{\Enc SM{e}, \Enc SM{i}}.
\end{equation}
This is essentially the first Futamura projection extended to abstract interpreters.

\section{Results}\label{sec:results}
\subsection{Correctness Theorem}\label{subsec:correctness}
Before we show the \hyperref[thm:correctness]{main theorem}, we prove the correctness of a meta-level analysis.

\begin{lem}[Correctness of a Meta-level Analysis]\label{lem:meta}
  Analyzing a \TGT*program using an abstract interpreter~$\Intp MS<\ABS>$ along with the concrete interpreter~$\Intp ST$ is sound.
\end{lem}
\begin{proof}
  Suppose we are given a \TGT*program~$\mathT{e}$ and a \TGT*value~$\mathT{i}$.
  Then by substituting $\Intp ST$ for $\mathS{e}$ and $\mathS{\Tup{\Enc TS{e}, \Enc TS{i}}}$ for $\mathS{i}$ in \cref{eq:aint}, we get
  \begin{equation*}
    \Conc S \Comp \Sem[\big]M{\Intp MS<\ABS>} \: \Tup[\big]{\Enc SM{\Intp ST}, \Enc SM{\Tup{\Enc TS{e}, \Enc TS{i}}}} \ni \Sem[\big]S{\Intp ST} \mathS{\Tup{\Enc TS{e}, \Enc TS{i}}}.
  \end{equation*}
  Then from \cref{eq:int}, $\Sem[\big]S{\Intp ST} \mathS{\Tup{\Enc TS{e}, \Enc TS{i}}}$ is precisely $\Enc TS{\Sem T{e}\: i}$.
  Thus the analysis~$\Sem[\big]M{\Intp MS<\ABS>} \: \Tup[\big]{\Enc SM{\Intp ST}, \Enc SM{\Tup{\Enc TS{e}, \Enc TS{i}}}}$ subsumes the concrete evaluation $\Enc TS{\Sem T{e}\: i}$.
\end{proof}

\correctness
\begin{proof}
  Given a partial evaluator~$\PE$ that satisfies \cref{eq:peval}, we can partially evaluate the abstract interpreter~$\Intp*MS<\ABS>$ with respect to the interpreter~$\Intp*ST$ as
  \begin{equation*}
    \Intp*MT<\ABS> \triangleq \Pe[\big]{\Intp*MS<\ABS>}{\Enc[\big]SM{\Intp*ST}}.
  \end{equation*}

  We now show that the specialized~$\Intp*MT<\ABS>$ is indeed a correct static analyzer of~\TGT.
  Fix a \TGT*program~$\mathT{e}$ and a \TGT*value~$\mathT{i}$.
  From \cref{eq:fut},
  \begin{equation*}
    \Sem[\big]M{\Intp*MT<\ABS>}\: \Enc SM{\Tup{\Enc TS{e}, \Enc TS{i}}} = \Sem[\big]M{\Pe[\big]{\Intp*MS<\ABS>}{\Enc SM{\Intp*ST}}}\: \Enc SM{\Tup{\Enc TS{e}, \Enc TS{i}}} = \Sem[\big]M{\Intp*MS<\ABS>}\: \Tup[\big]{\Enc SM{\Intp*ST}, \Enc SM{\Tup{\Enc TS{e}, \Enc TS{i}}}}
  \end{equation*}
  and from \cref{lem:meta}, this soundly approximates the evaluation
  \begin{equation*}
    \Enc TS{\Sem T{e}\: i} \in \Conc S \Comp \Sem[\big]M{\Intp*MT<\ABS>}\: \Enc SM{\Tup{\Enc TS{e}, \Enc TS{i}}}.
  \end{equation*}
  Therefore, the specialized abstract interpreter~$\Intp*MT<\ABS>$ is correct.
\end{proof}

\subsection{A Mini-Example}\label{subsec:example}
We now present a concrete mini-example that demonstrates how our approach retargets an abstract interpreter~$\Intp*MS<\ABS>$ for a source language~\SRC{} to an abstract interpreter~$\Intp*MT<\ABS>$ for a target language~\TGT.
This example illustrates how our method works in practice and highlights the properties of the resulting retargeted abstract interpreter.

Consider an ML-like functional language as the meta-language~\MET{} and loop-free mini-languages~\TGT{} and~\SRC:
\begin{center}
  \small
  \begin{bnf}
    \mathT{e} : $\Dom TE$ ::= \Tadd{n} // \Tmult{n}
    ;;
    \mathS{e} : $\Dom SE$ ::= $\mathS{x}$ // $\mathS{n}$ // \Sadd{e}{e} // \Smult{e}{e} // \Sequ{e}{e} // \Spair{e}{e} // \Sfst{e} // \Ssnd{e}
    | \Scond{e}{e}{e}
  \end{bnf}
\end{center}
\TGT{}~is a set of single instructions that either add or multiply a number to a user input, e.g., \Tadd{42} adds 42 to an input and \Tmult{42} multiplies an input by 42.
\SRC{}~is a structured language with a single variable~$\mathS{x}$ that represents the user input, e.g., \Sfst{\mathS{x}} returns the first element of a (tuple) input.

\IncludeEnc
\IncludeDec
Given the value domains $\Dom TV = \mathbb{Z}$ and $\Dom SV = \mathbb{Z} + \Dom SV \times \Dom SV$,
we have the semantics~$\Sem T{e}$ and~$\Sem S{e}$ and the encoder~$\Enc TS{\cdot}$ as follows:
\begin{gather*}
  \small
  \begin{flalign*}
    \Sem T{e}         & \in \Dom TV \domto \Dom TV     & \Enc TS{e}         & \in \Dom SV                          & \Enc TS{v} & \in \Dom SV \\
    \Sem T{\Tadd{n}}  & = \func{v}{\mathT{n} + v}      & \Enc TS{\Tadd{n}}  & = \mathS{\Tup{\Plain{0}, \mathT{n}}} & \Enc TS{n} & = \mathT{n} \\
    \Sem T{\Tmult{n}} & = \func{v}{\mathT{n} \times v} & \Enc TS{\Tmult{n}} & = \mathS{\Tup{\Plain{1}, \mathT{n}}}
  \end{flalign*} \\
  \begin{flalign*}
    \Sem S{e} & \in \Dom SV \domto \Dom SV & \Sem S{x} & = \func{v}{v} & \Sem S{n} & = \func{v}{\mathS{n}}
  \end{flalign*} \\
  \begin{flalign*}
    \Sem S{\Sadd{e_1}{e_2}} & = \func{v}{\Sem S{e_1}\: v + \Sem S{e_2}\: v} & \Sem S{\Smult{e_1}{e_2}} & = \func{v}{\Sem S{e_1}\: v \times \Sem S{e_2}\: v}                                  \\
    \Sem S{\Sequ{e_1}{e_2}} & = \func{v}{\Sem S{e_1}\: v = \Sem S{e_2}\: v} & \Sem S{\Spair{e_1}{e_2}} & = \func{v}{\mathS{\Tup{\Sem S{e_1}\: \Plain{v}, \Sem S{e_2}\: \Plain{v}}}} \\
    \Sem S{\Sfst{e}}        & = \func{v}{\pi_1 (\Sem S{e}\: v)}             & \Sem S{\Ssnd{e}}         & = \func{v}{\pi_2 (\Sem S{e}\: v)}
  \end{flalign*} \\
  \Sem S{\Scond{e_p}{e_c}{e_a}} = \func{v}{\scond{\Sem S{e_p}\: v \ne 0}{\Sem S{e_c}\: v}{\Sem S{e_a}\: v}}
\end{gather*}
Here, $\pi_1$ and $\pi_2$ are the projections of a pair, and $\scond{\cdot}{\cdot}{\cdot}$ is the semantic conditional operator.
We omit the presentation of~$\Dom ME$,~$\Dom MV$,~$\Enc SM{\cdot}$, and~$\Sem M{e}$ of~\MET{} for brevity.
The encoder~$\Enc SM{\cdot}$ is a typical parser that embeds the AST of an \SRC*program or an \SRC*value into its corresponding \MET*datatype representation.

Now we can define the concrete \TGT*interpreter
\begin{equation*}
  \Intp*ST = \Scond{\Sequ{\Sfst{\Sfst{x}}}{0}}{\Sadd{\Ssnd{\Sfst{x}}}{\Ssnd{x}}}{\Smult{\Ssnd{\Sfst{x}}}{\Ssnd{x}}}
\end{equation*}
where we encode \TGT*program~$\mathT{e}$ and \TGT*value~$\mathT{i}$ as $\mathS{x} \coloneq \mathS{\Tup{\Enc TS{e}, \Enc TS{i}}}$.
Note that we match with opcodes from the encoding~$\Enc TS{e}$.

The abstract interpreter~$\Intp*MS<\ABS>$ parameterized by the base value abstraction is then given by (encodings~$\Enc SM{e}$ and~$\Enc SM{v}$ are implicit)
\begin{metcode}
let $\mathM{\Intp*MS<\ABS>(\textS{e}, \textS{i})}$ = let $\mathM{i^\ABS}$ = $\mathM{\eta(\textS{i})}$ in $\mathM{\textM{eval}^\ABS(\textS{e}, i^\ABS)}$
where $\mathM{\textM{eval}^\ABS(\textS{e}, v^\ABS)}$ = match $\textS{e}$ with
  | $\textS{x} \to \mathM{v^\ABS}$ | $\textS{n} \to \eta(\textS{n})$
  | $\mathM{\Sadd{e_1}{e_2} \to \textM{eval}^\ABS(\mathS{e_1}, v^\ABS) \mathbin{+^\ABS} \textM{eval}^\ABS(\mathS{e_2}, v^\ABS)}$
  | $\Smult{e_1}{e_2} \to \mathM{\textM{eval}^\ABS(\mathS{e_1}, v^\ABS) \mathbin{\times^\ABS} \textM{eval}^\ABS(\mathS{e_2}, v^\ABS)}$
  | $\Sequ{e_1}{e_2} \to \mathM{\textM{eval}^\ABS(\mathS{e_1}, v^\ABS) \mathbin{=^\ABS} \textM{eval}^\ABS(\mathS{e_2}, v^\ABS)}$
  | $\Spair{e_1}{e_2} \to \mathM{(\textM{eval}^\ABS(\mathS{e_1}, v^\ABS), \textM{eval}^\ABS(\mathS{e_2}, v^\ABS))}$
  | $\Sfst{e} \to \mathM{\pi_1(\textM{eval}^\ABS(\mathS{e}, v^\ABS))}$ | $\Ssnd{e} \to \mathM{\pi_2(\textM{eval}^\ABS(\mathS{e}, v^\ABS))}$
  | $\Scond{e_p}{e_c}{e_a} \to$ let $\mathM{p^\ABS}$ = $\mathM{\textM{eval}^\ABS(\mathS{e_p}, v^\ABS)}$ in
    $\mathM{\mathcal{F}^\ABS_{\ne 0}(p^\ABS, \textM{eval}^\ABS(\mathS{e_c}, v^\ABS)) \mathbin{\sqcup^\ABS} \mathcal{F}^\ABS_{= 0}(p^\ABS, \textM{eval}^\ABS(\mathS{e_a}, v^\ABS))}$
\end{metcode}
The parametrizable knob is given by the extraction function $\eta$ \citep{NieNieHan99Principles,DarHor19Constructive} that extracts the abstract value from a single concrete value.
Abstract operators are defined correspondingly: abstract operators~$+^\ABS$,~$\times^\ABS$,~$=^\ABS$ correspond to~$+$,~$\times$,~$=$, respectively; $\sqcup^\ABS$~is an abstract join; and~$\mathcal{F}^\ABS_{\ne 0}$ and~$\mathcal{F}^\ABS_{= 0}$ filter abstract values based on the predicate~$\ne 0$ and~$= 0$, respectively \citep{RivYi20Introduction}.
We preserve the tuple structure of the values and use the concrete projections~$\pi_1$ and~$\pi_2$ and concrete tuple constructor~$\mathM{(\cdot, \cdot)}$.
\OmitEnc
\OmitDec

Performing the specialization by hand, we get the retargeted abstract interpreter~$\Intp MT<\ABS> = \Pe[\big]{\Intp MS<\ABS>}{\Enc[\big]SM{\Intp ST}}$ as
\begin{metcode}
let $\mathM{\Intp*MS<\ABS>_{\Intp ST}(\textS{i})}$ =
  let $\mathM{i^\ABS}$ = $\mathM{\eta(\textS{i})}$ in let $\mathM{p^\ABS}$ = $\mathM{\pi_1(\pi_1(i^\ABS)) \mathbin{=^\ABS} \eta(0)}$ in
  $\mathM{\mathcal{F}^\ABS_{\ne 0}(p^\ABS, \pi_2(\pi_1(i^\ABS)) \mathbin{+^\ABS} \pi_2(i^\ABS)) \mathbin{\sqcup^\ABS} \mathcal{F}^\ABS_{= 0}(p^\ABS, \pi_2(\pi_1(i^\ABS)) \mathbin{\times^\ABS} \pi_2(i^\ABS)))}$
\end{metcode}
where the subscript~$\Intp ST$ indicates that the source code~$\Intp MS<\ABS>$ is specialized with respect to the concrete interpreter~$\Intp ST$.

Notice that while the retargeted abstract interpreter~$\Intp*MT<\ABS>$ is written in the meta-language~\MET, the (abstract-)interpretative overhead is eliminated and the structure of the interpreter~$\Intp*ST$ is closely mirrored.
This is a classical phenomenon in partial evaluation \citep{Jon96Introduction}.

Moreover, observe the key advantage of our approach: $\Intp*MT<\ABS>$ reuses the abstract operators~$+^\ABS$ and~$\times^\ABS$ from the existing~$\Intp*MS<\ABS>$ to analyze the \TGT*programs~\Tadd{n} and~\Tmult{n}, resulting in a sound analyzer (\hyperref[thm:correctness]{main theorem}).
Our approach enjoys the inherited properties from the existing (abstract) interpreter, as observed by \citet{Rey72Definitional}.

\section{Future Work}\label{sec:future}
Our recipe for retargeting abstract interpreters for new languages opens several promising research directions.

To evaluate our approach, we are working on an expanded example from \cref{subsec:example} where~\SRC{} and~\TGT{} include loops.

We plan to investigate the following research questions:
\begin{description}
  \item[RQ1] How does a retargeted abstract interpreter compare to manually crafted ones and meta-level analyzers in terms of precision and performance (in terms of time complexity)?
  \item[RQ2] How do semantic gaps between source and target languages (e.g., value domains, higher- vs. first-order) affect the retargeted abstract interpreter?
  \item[RQ3] How does the partial evaluation strategy impact the quality of the retargeted abstract interpreter?
  \item[RQ4] How can this technique be extended to analyze real-world programs with more complex abstractions?
\end{description}

\begin{acks}
  I would like to express my deep gratitude to my advisor, Kwangkeun Yi, for his invaluable guidance throughout this research.
  I am also grateful to Joongwon Ahn and Joonhyup Lee for their helpful discussions and suggestions.
  I thank the reviewers and shepherd for their valuable feedback.
\end{acks}

\bibliographystyle{ACM-Reference-Format}
\bibliography{references}
\end{document}